\documentclass{article}
\usepackage[longend,ruled,english,linesnumbered]{algorithm2e} 
\usepackage[utf8]{inputenc}
\usepackage{mathtools}
\usepackage{comment}
\usepackage{amsmath}
\usepackage{algpseudocode}
\usepackage{amsfonts,amssymb,amsthm,boxedminipage,color,url,fullpage}
\usepackage{amsfonts}
\usepackage{bbm}
\usepackage{amssymb}
\usepackage{amsthm}
\usepackage{boxedminipage}
\usepackage{color}
\usepackage{url}
\usepackage{fullpage}
\usepackage{mathtools}
\usepackage[numbers]{natbib}

\usepackage{thmtools} 
\usepackage{thm-restate}

\usepackage{enumitem}
\usepackage{tcolorbox}
\usepackage[dvipsnames]{xcolor}
\usepackage[colorlinks=true,citecolor=blue,linkcolor=blue,urlcolor=blue]{hyperref}

\usepackage[labelfont=bf]{caption}
\usepackage{aliascnt,cleveref}
\usepackage{authblk}
\usepackage{accents}
\usepackage{tikz}
\usetikzlibrary{positioning,arrows}
\usepackage{pgfplots}
\pgfplotsset{width=10cm,compat=1.9}
\usepgfplotslibrary{external}
\allowdisplaybreaks

\usepackage[]{color-edits}

\newcommand{\ceil}[1]{\lceil #1 \rceil}

\newcommand{\bone}{\mathbbm{1}}
\newcommand{\MMS}{\mu}
\newcommand{\Normal}[1]{\bar #1}

\newcommand{\ti}[1]{\textcolor{Green}{(Timo: #1)}}

\newtheorem{example}{Example}
\newtheorem{theorem}{Theorem}
\newtheorem{definition}{Definition}
\newtheorem{lemma}{Lemma}
\newtheorem{corollary}{Corollary}
\newtheorem{remark}{Remark}

\title{Simultaneous Ordinal Maximin Share and Envy-Based Guarantees}

\author{
\quad Hannaneh Akrami\thanks{Max Planck Institute for Informatics and Universität des Saarlandes. Email: \texttt{hakrami@mpi-inf.mpg.de}}\thanks{Hertz Chair for Algorithms and Optimization, University of Bonn.}
\quad
Timo Reichert\thanks{University of Bonn. Email: \texttt{s6tireic@uni-bonn.de }} 
}
\date{}

\begin{document}

\maketitle

\begin{abstract}
    We study the fair allocation of indivisible goods among agents with additive valuations. The fair division literature has traditionally focused on two broad classes of fairness notions: envy-based notions and share-based notions. Within the share-based framework, most attention has been devoted to the maximin share (MMS) guarantee and its relaxations, while envy-based fairness has primarily centered on EFX and its relaxations. Recent work has shown the existence of allocations that simultaneously satisfy multiplicative approximate MMS and envy-based guarantees such as EF1 or EFX. 

    Motivated by this line of research, we study—for the first time—the compatibility between \emph{ordinal} approximations of MMS and envy-based fairness notions. In particular, we establish the existence of allocations satisfying the following combined guarantees: (i) simultaneous $1$-out-of-$\ceil{3n/2}$ MMS and EFX for ordered instances; (ii) simultaneous $1$-out-of-$\ceil{3n/2}$ MMS and EF1 for top-$n$ instances; and (iii) simultaneous $1$-out-of-$4\ceil{n/3}$ MMS and 
    EF1 for ordered instances.
\end{abstract}

\section{Introduction}

Fair allocation of resources is a fundamental problem at the intersection of computer science, economics, and social choice theory, and has received significant attention since the seminal work of \cite{steinhaus1948problem}. The goal is to divide a set $M$ of $m$ items among a set $N$ of $n$ agents, where each agent $i \in N$ has a valuation function $v_i: 2^M \rightarrow \mathbb{R}_{\geq 0}$ over subsets of items. Over the years, a rich collection of fairness notions has been proposed, which can broadly be categorized into \emph{share-based} and \emph{envy-based} notions. Among these, proportionality (and its relaxations) and envy-freeness (and its relaxations) are widely regarded as the canonical representatives of share-based and envy-based fairness, respectively. In this work, we focus on the fair allocation of \emph{indivisible} goods among agents with additive valuations.

In the discrete fair division setting, where resources consist of indivisible goods, proportionality is often unattainable. A central relaxation of proportionality in this setting is the \emph{maximin share} (MMS) guarantee \cite{budish2011combinatorial}. An allocation \(A = (A_1, \ldots, A_n)\) is a partition of \(M\) into \(n\) bundles, where each agent receives exactly one bundle. MMS of each agent $i$ is the maximum value she can obtain by receiving the minimum valued bundle in any allocation. Formally, let $\Pi_n$ be the set of all possible partitions of $M$ into $n$ bundles. Then
$$\MMS^n_i(M) = \max_{A \in \Pi_n} \min_j v_i(A_j).$$
An allocation is MMS if all agents value their bundle at least as much as their MMS. 

It is known that exact MMS allocations do not always exist \cite{kurokawa2018fair}, which has led to an extensive line of work on approximate MMS guarantees. Formally, an allocation $X$ is $\alpha$-MMS, iff $v_i(X_i) \geq \alpha \cdot \mu^n_i(M)$ for all agents $i$. A long sequence of results has progressively improved the approximation factor, with the current state-of-the-art being the existence of \(\alpha\)-MMS allocations for \(\alpha = 7/9\) \cite{huang2025fptas79}.

On the envy-based side, \emph{envy-freeness up to any item} (EFX) has emerged as one of the most prominent relaxations of envy-freeness for indivisible goods. An allocation \(A\) is said to be EFX if for all agents \(i,j \in N\),
\[
v_i(A_i) \geq v_i(A_j \setminus \{g\}) \quad \text{for all } g \in A_j.
\]
In contrast to MMS, it remains an open question whether EFX allocations are guaranteed to exist. Nevertheless, a substantial body of work has studied relaxations of EFX, as well as the existence of EFX allocations in more restricted settings (see Section \ref{sec:related-work}). Another related line of research considers \emph{partial allocations}, in which some items may remain unallocated \cite{caragiannis2019envy}. In this context, combining EFX guarantee with efficiency notions is crucial in order to avoid highly wasteful outcomes.

While much of the fair division literature focuses on individual fairness notions in isolation, \cite{amanatidis2020multiple} initiated the study of \emph{simultaneously satisfying share-based and envy-based fairness guarantees} under additive valuations. In particular, they established the existence of allocations that are both \(0.553\)-MMS and \(0.618\)-EFX. Subsequent work strengthened this line of results: \cite{chaudhury2021little} proved the existence of partial allocations that are simultaneously \(1/2\)-MMS and EFX, which was later improved to \(2/3\)-MMS and EFX by \cite{AkramiRathi2025simultaneous}. For additive valuations, \cite{AshuriMXS-EFL} prove the existence of complete allocations that are MXS (a share-based notion weaker than MMS) and EFL (an envy-based notion weaker than EFX) at the same time. Very recently, this was improved to EEFX and EFL by \cite{akrami2026achievingef1epistemicefx}. \cite{feige2025residualmaximinshare}, introduced the \emph{residual MMS (RMMS)} notion and proved there exists partial allocations that are RMMS and EFX, and complete allocations that are both RMMS and EFL. 

Most of the aforementioned results combine \emph{multiplicative (cardinal) approximations} of MMS with EFX or its relaxations. Another important and well-studied class of relaxations of MMS is given by \emph{ordinal} approximations which was introduced by \cite{budish2011combinatorial} together with the original definition of MMS. In this relaxation, instead of partitioning the goods into $n$ bundles to obtain agents' desired threshold (MMS), we partition them into $d > n$ bundles. Similarly, we say an allocation $A$ is \emph{$1$-out-of-$d$ MMS}, iff $v_i(A_i) \geq \MMS_i^d(M)$. By setting $d=n$, we get the exact MMS. Ordinal MMS guarantees depend only on the relative ranking of the subsets of goods rather than their exact values \cite{Hosseini2021OrdinalMS}. \cite{budish2011combinatorial} demonstrated that a $1$-out-of-$(n+1)$ MMS allocation can be achieved when \emph{excess goods} are allowed. Absent such excess goods, only weaker guarantees are known. The earliest non-trivial result in the standard setting established the existence of $1$-out-of-$d$ MMS allocations with $d = 2n - 2$~\cite{AignerHorev2019EnvyfreeMI}. Subsequent work progressively tightened this bound, first to $d = \lceil 3n/2 \rceil$~\cite{hosseini2021guaranteeing}, then to $d = \lfloor 3n/2 \rfloor$~\cite{Hosseini2021OrdinalMS}, culminating in the best-known guarantee of $d = 4\ceil{n/3}$ \cite{MMS-l-out-of-d}. Despite these advances, whether a $1$-out-of-$(n+1)$ MMS allocation exists in the standard setting remains unresolved.

\subsection{Our Results}
In this work, we initiate the study of the \emph{compatibility between ordinal relaxations of MMS and envy-based fairness notions}, such as EFX and its relaxations. To the best of our knowledge, this is the first systematic investigation of how ordinal share-based guarantees interact with envy-based notions in the allocation of indivisible goods. 

We primarily focus on \emph{ordered instances}. In such instances, all agents agree on the relative ordering of goods; that is, there exists a renaming of the goods as \(g_1, \ldots, g_m\) such that for all $i \in N$,
$v_i(g_1) \geq \cdots \geq v_i(g_m)$.

For share-based notions, it is known that any instance can be reduced to an ordered instance without loss of generality \cite{barman2020approximation}. Consequently, when studying share-based guarantees alone, one may assume that the instance is ordered. However, this reduction does not preserve envy-based properties. While a simple envy-cycle elimination algorithm \cite{lipton2004approximately}, applied by allocating goods in decreasing order of value, guarantees EFX for ordered instances, the existence of EFX allocations for general (non-ordered) instances remains one of the most important open problems in fair division.

Our first main result establishes the existence of \emph{complete} allocations—where all goods are allocated—that are simultaneously EFX and \(1\)-out-of-\(\lceil 3n/2 \rceil\) MMS for ordered instances.
\begin{restatable}{theorem}{theoremOne}\label{thm:1}
    For every ordered instance $(M,N,(v_i)_{i\in N})$ of fair division, there exists a complete allocation that is 1-out-of-$\ceil{3n/2}$ MMS and EFX at the same time.
\end{restatable}

Our second main result extends the compatibility of \(1\)-out-of-\(\lceil 3n/2 \rceil\) MMS and envy-based fairness to a more general setting, which we call \emph{top-\(n\) instances}. Top-$n$ instances have recently received heightened attention. They were introduced in \cite{Markakis2023_topN}, and further expanded upon in \cite{Mancho2025_topN}. In this setting, agents agree only on the set of the top \(n\) most valuable goods; that is, there exists a renaming of the goods \(g_1, \ldots, g_m\) such that
\[
v_i(g_j) \geq v_i(g_{j'}) \quad \text{for all } i \in N,\ j \in [n],\ j' \in [m] \setminus [n].
\]
This assumption is strictly weaker than full orderedness. Under this setting, we prove the existence of complete allocations that are simultaneously EF1 and \(1\)-out-of-\(\lceil 3n/2 \rceil\) MMS. 
\begin{restatable}{theorem}{theoremTwo}
    Given a top-$n$ instance $I$, there exists a partial allocation that is both 1-out-of-$\ceil{3n/2}$ MMS and EFX, and there exists a complete allocation that is both 1-out-of-$\ceil{3n/2}$ MMS and EF1.    
\end{restatable}

Finally, we establish the compatibility of weaker envy-based notions with stronger ordinal MMS guarantees for ordered instances:
\begin{restatable}{theorem}{theoremThree}
    Given an ordered instance $I$, there exists a complete allocation that is both $1$-out-of-$4\ceil{n/3}$ MMS and EF1.
\end{restatable}

We conclude by briefly comparing our results with analogous work on combining \emph{cardinal} MMS approximations with envy-based guarantees. While existing results for cardinal MMS and EFX (or its relaxations) apply to general instances \cite{AkramiRathi2025simultaneous}, the best-known guarantees for ordered instances are no stronger. In particular, to the best of our knowledge, the largest value of $\alpha$ for which $\alpha$-MMS can be guaranteed together with EFX or EF1—even for ordered instances—is $\alpha=2/3$, despite the fact that the state-of-the-art MMS approximation alone achieves $\alpha=7/9$ \cite{huang2025fptas79}. 

For ordered instances our result attains the state-of-the-art \emph{ordinal} MMS guarantee of $1$-out-of-$4\ceil{n/3}$ while simultaneously satisfying EF1. Moreover, prior results combining $\alpha$-MMS with EFX apply only to \emph{partial} allocations, whereas Theorem~\ref{thm:1} guarantees complete allocations that satisfy EFX.

\subsection{Technical Contribution}


The main technical challenges addressed in this paper stem from the fact that existing algorithms for \(1\)-out-of-\(d\) MMS are largely agnostic to the structural properties required to ensure envy-freeness. Consequently, these algorithms rely on assumptions that cease to be without loss of generality once envy-based guarantees are introduced. Two such assumptions are \emph{orderedness} and \emph{normalization} (see Definition \ref{def:normal}) of instances. 

To relax the former, we introduce the notion of \emph{top-\(n\)} instances. Although this property is not without loss of generality, it constitutes a substantially weaker requirement than full orderedness and significantly broadens the class of instances that can be handled by Algorithm~\ref{3/2top:alg:lone_divider}. 

To address the normalization assumption, we propose a novel normalization technique that preserves any given ordering of valuations; see Lemma~\ref{4/3ordered:lem:normalization}. This construction allows us to decouple the valuations used by Algorithm~\ref{4/3ordered:alg:main} from those employed in the analysis. As a result, we retain the analytical advantages of working with normalized instances while ensuring that the algorithm operates on the original valuations, which is crucial for guaranteeing EF1. 

Finally, the only previously known algorithm achieving $1$-out-of-$4\lceil n/3\rceil$ MMS~\cite{MMS-l-out-of-d} relies on a bag-filling process that allocates bags in a fixed order and fills them according to a prescribed item order. These properties are incompatible with our setting, as allowing agents who have already received a bag to swap bundles precludes maintaining either order. Consequently, the existing analysis cannot be applied directly. We therefore develop a modified analysis that decouples the performance guarantees from these ordering assumptions, enabling correctness under the additional flexibility required for envy-freeness.


\subsection{Further Related Work}\label{sec:related-work}
In the context of the EFX criterion, \cite{plaut2020almost} established existence for the case of two agents with monotone valuations. 
For the case of three agents, \cite{chaudhury2020efx} proved that EFX allocations always exist for three agents when valuations are additive. This was extended to broader valuation classes in a sequence of works~\cite{chaudhury2020efx,berger2021almost,AkramiACGMM23} and to three types of additive valuations \cite{VishwaEFX}. Existence has also been shown in several restricted valuation domains, including identical~\cite{plaut2020almost}, binary~\cite{halpern2020fair}, and bi-valued~\cite{amanatidis2021maximum} preferences.

In parallel, researchers have explored both approximate notions~\cite{chaudhury2021little,amanatidis2020multiple,chan2019maximin,farhadi2021almost} and weakened variants~\cite{amanatidis2021maximum,caragiannis2019envy,berger2021almost,mahara2021extension,chasmjahan23,aram22,ef2x} of EFX, making this an active area within discrete fair division. Motivated by the epistemic perspective introduced by \cite{ABCGL18}, \cite{Caragiannis2023} proposed a novel relaxation known as \emph{epistemic EFX} (EEFX). It has been shown that EEFX allocations always exist under monotone valuations~\cite{AR24eefx} and admit polynomial-time algorithms in the additive case~\cite{Caragiannis2023}.

More relevant to our setting, Markakis and Santorinaios \cite{Markakis2023_topN} show that $2/3$-EFX allocations can be guaranteed for top-$n$ instances. Additionally, a relaxation of EFX—termed envy-freeness up to any flip—is introduced and studied in this setting by Mancho, Markakis, and Protopapas \cite{Mancho2025_topN}.

Turning to maximin share, since the impossibility result by \cite{procaccia2014fair}, a substantial body of work has focused on achieving multiplicative approximations of MMS \cite{amanatidis2017approximation,kurokawa2018fair,ghodsi2018fair,barman2020approximation,garg2020improved,FST21,simple,akrami2024breaking,heidari2025improved1013} for additive valuations leading to $7/9$ \cite{huang2025fptas79}.

For a broader treatment of these results and related work, we refer the reader to the survey by~\cite{survey2022}. Closely related is the extensively studied problem of allocating \emph{chores} rather than goods; see~\cite{guo2023survey} for a comprehensive overview.

\section{Definitions and Notation}\label{sec2}
We denote an instance of fair division by $I=(M,N,(v_i)_{i\in N})$, where $M$ denotes the $m$ indivisible goods, and $N$ the $n$ agents. For every agent, $v_i:2^M\to\mathbb R_{\geq0}$ denotes an additive valuation function, i.e. $v_i(X)=\sum_{g\in X}v_i(\{g\})$. For ease of notation, we write $v_i(g)$ instead of $v_i(\{g\})$ for goods $g\in M$. Without loss of generality, we can write $M=[m]$ and $N=[n]$, where $[k]=\{1,\dots,k\}$ for all integers $k$. 

\begin{definition}
    The \emph{1-out-of-$d$ MMS} value of an agent $i\in N$ is defined by 
    $$\MMS_i^d(M):=\max_{A\in \Pi_d(M)}\min_{A_j\in A}v_i(A_j),$$
    where $\Pi_d(M)$ denotes the partitions of $M$ into $d$ subsets. If $M$ is clear from context, we omit it and write $\MMS_i^d$. If $d=n$, we simply call this the \emph{MMS} value and write $\MMS_i$.
\end{definition}

A partition $P$ of $M$ into $d$ bundles is a $\MMS^d(M)$ partition for agent $i$, if $\min_{j \in d} v_i(P_j) = \MMS_i^d(M)$.

\begin{definition}
    An allocation $A$ is called \emph{1-out-of-$d$ MMS}, if for all $i \in N$, 
    $v_i(A_i)\geq \MMS_i^d(M)$.

\end{definition}

\begin{definition}
    Given a bundle $B$ and a (partial) allocation $(A_1,\dots, A_n)$, we say $i$ strongly envies $B$ if there exists $g\in B$ such that $v_i(A_i)<v_i(B\setminus\{g\})$.
\end{definition}

\begin{definition}
    An allocation $A$ is called \emph{envy-free up to any good}, or \emph{EFX}, if for any pair of agents $i\neq j$ and all $g \in A_j$, we have 
    $v_i(A_i)\geq v_i(A_j\setminus\{g\})$.
    In other words, no agent strongly envies another agent.


    Similarly, $A$ is called \emph{envy-free up to one good}, or \emph{EF1}, if for any pair of agents $i\neq j$ we have $v_i(A_i)\geq v_i(A_j)$, or there exists $g\in A_j$ with
    $v_i(A_i)\geq v_i(A_j\setminus\{g\})$.
\end{definition}

\begin{definition}
    Given a bundle $B$, and a (partial) allocation $A$, an agent $i$ is called \emph{most envious} of $B$, if there exists a proper subset $B'\subsetneq B$ such that $v_i(B')>v_i(A_i)$, and no other agent strongly envies $B'$.
\end{definition}
If there exists an agent who strongly envies $B$, there is also a most envious agent of $B$. Namely 
let $B'$ be a minimal subset of $B$ envied by some agent. An agent $i$ who envies $B'$ is a most envious agent.

\begin{definition}
    An instance $I=(M,N,(v_i)_{i\in N})$ is called \emph{ordered} if there exists an ordering $g_1,\dots, g_m$ of $M$ such that 
    $v_i(g_1)\geq v_i(g_2)\geq \dots\geq v_i(g_m)$ 
    for all agents $i\in N$.
\end{definition}

\begin{definition}\label{def:topK}
    An instance $I=(M,N,(v_i)_{i\in N})$ is a \emph{top-$k$} instance for $k\in[m]$, if the set of the $k$ most valuable goods is identical for all agents. We call goods in this set the \emph{top $k$} goods of $I$.
\end{definition}
Note that the agents need not agree on the values, or even the order among the items in the \emph{top $k$} goods of $I$. 

\begin{remark}
    Definition \ref{def:topK} is weaker than orderedness. In particular, every ordered instance is top-$k$ for any $k\in[m]$. Hence, for share based fairness, any instance can be assumed to be top-$k$ without loss of generality. For envy based notions however, 
    no such reduction is known.
\end{remark}

All algorithms presented in this paper only compute partial allocations at first, i.e. allocations where not all goods are necessarily assigned to agents. If we only wish to achieve an MMS-based guarantee, we can assign the remaining goods to agents arbitrarily. However, this does not work for envy-based guarantees. To remedy this, we employ the \emph{envy-cycle elimination algorithm}, first presented in \cite{lipton2004approximately}. Briefly summarized, the algorithm successively adds the remaining goods to existing bags in the following way. Consider the digraph $G$ where the nodes are agents, and there exists an edge $(i,j)$ if $i$ envies $j$. If $G$ contains an agent $i$ with no incoming edge, we can add the most valuable (for $i$) unassigned item to $i$'s bundle. If $G$ contains no such agent, then it must contain a directed cycle $C$. In this case, we can assign each agent in $C$ the bundle of the next agent along $C$. We eliminate cycles this way, until we find an agent with no incoming edge to assign the next good to. For a full description of the algorithm, see Algorithm \ref{common:alg:envycycle} in Appendix \ref{app2}. This procedure leads to the following lemma:

\begin{restatable}{lemma}{lemEnvyCycle}\label{3/2ordered:lem:envy_cycles}
    Let $I$ be an instance, and $A$ be a partial allocation that is EF1. Then there exists a complete allocation $A'$ that is also EF1, and $v_i(A_i)\leq v_i(A_i')$ for all agents $i$.
    
    If additionally $I$ is ordered, $A$ is EFX, and 
    each allocated good is at least as valuable as each unallocated good, there exists a complete allocation $A'$ that is also EFX, and $v_i(A_i)\leq v_i(A_i')$ for all agents $i$.
\end{restatable}

Since all valuation functions we consider here are additive, we can divide $v_i(g)$ by $\MMS_i^d$ for all $i\in N, g\in M$, and receive an equivalent instance with $\MMS_i^d=1$ for all agents.  
Note that all the envy-based notions that we study are scale-invariant. Hence, without loss of generality, we assume $\MMS^d_i = 1$ for $d=\ceil{3n/2}$ in Sections \ref{sec:3/2ordered} and \ref{sec4}, and $d=4\ceil{n/3}$ in Section \ref{sec5}.

\section{\boldmath $1$-out-of-$\ceil{3n/2}$ MMS Together with EFX for Ordered Instances}\label{sec:3/2ordered}
In this section, we present an algorithm that gets an ordered instance and outputs an allocation that is $1$-out-of-$\ceil{3n/2}$ MMS and EFX at the same time (Algorithm \ref{3/2ordered:alg:main}). 

Algorithm~\ref{3/2ordered:alg:main} follows the approach of \cite{MMS-l-out-of-d}, but is augmented with an additional phase that allows agents to claim a bundle even if they have already been allocated one. As in \cite{AkramiRathi2025simultaneous}, this extra phase mitigates the myopic behavior that agents would otherwise exhibit. A more subtle yet crucial departure from \cite{MMS-l-out-of-d} lies in the initialization of the bags: unlike their approach, we permit bags of size one, a modification that is essential for guaranteeing the EFX property.

\begin{algorithm}[tb]
    \DontPrintSemicolon

    \SetKwInOut{Input}{Input}
    \SetKwInOut{Output}{Output}
    \Input{An ordered instance $(M,N,(v_i)_{i\in N})$}
    \Output{A partial allocation that is 1-out-of-$\ceil{3n/2}$ MMS and EFX}
    \BlankLine
    $k\gets1$,\quad $U\gets N$\;
    \While(\tcp*[f]{Initialization}){$\exists i\in U:v_i(k)\geq 1$}{
        $B_k\gets \{k\}$\;
        $A_i\gets B_k$,\quad $U\gets U\setminus\{i\}$\;
        $k\gets k+1$\;
    }
    \For{$j$ from $k$ to $n$}{
        $B_j\gets \{j,2n-j+1\}$\;
    }
    $g\gets 2n+2-k$,\quad $L\gets \{B_1,\dots,B_n\}$\;
    \While(\tcp*[f]{Bag-Filling}){$U\neq \emptyset$}{
        \uIf{$\exists i\in U, B\in L:v_i(B)\geq 1$}{
            $L\gets L\setminus\{B\}$,\quad $U\gets U\setminus\{i\}$\;
            $A_i\gets B$\;
        }\uElseIf{$\exists i\in N\setminus U, B\in L:v_i(B)>v_i(A_i)$}{
            $L\gets (L\cup\{A_i\})\setminus \{B\}$\;
            $A_i\gets B$\;
        }\Else{
            Let $B\in L$ arbitrary: $B\gets B\cup\{g\}$\;
            $g\gets g+1$\;
        }
    }
    \Return $(A_1,\dots ,A_n)$\;
    \caption{\sc{1-out-of $3n/2$ MMS + EFX}}
    \label{3/2ordered:alg:main}
\end{algorithm}

We can assume without loss of generality, that $m\geq2n$.\footnote{Otherwise, we add dummy goods of value $0$ to every agent, until this is true.} Now consider the bag filling approach shown in Algorithm \ref{3/2ordered:alg:main}: we first initialize $n$ bags, such that bag $B_j$ contains good $j$, and possibly also good $2n-j+1$. Then, we successively add the most valuable leftover good to an arbitrary bag, and allow agents without a bag to claim bags they value at least 1, and agents with a bag to swap for a bag they value more than their current bag. The latter is crucial to ensure EFX. 

We first prove that Algorithm \ref{3/2ordered:alg:main} produces a partial allocation that is 1-out-of-$\lceil 3n/2\rceil$ MMS, and later that the allocation is also EFX. Since we only assign bags to agents who value them at least 1, if the algorithm allocates a bag to every agent, the allocation is 1-out-of-$\lceil 3n/2\rceil$ MMS. Towards a contradiction, let us assume that there exists an agent $i$ that does not receive a bag in this process. We call the bag that results from $B_j$ at the end of the algorithm $\Hat{B_j}$.
\begin{lemma}\label{3/2ordered:lem:leftover}
    For each bag $\Hat B_j$, let $g_j$ be the second most valuable good in $\Hat B_j$, if it exists (and $\emptyset$ otherwise). We then have
    $\sum_{j=1}^nv_i(\Hat B_j\setminus\{g_j\})<\sum_{j=1}^n\max\{1,v_i(j)\}.$
\end{lemma}

\begin{proof}
    If no items get added to $B_j$ after initialization, we have $v_i(\Hat B_j\setminus\{g_j\})=v_i(j)$. Otherwise, let $g$ be the last item added to $\Hat B_j$. Then by orderedness, $v_i(g)\leq v_i(g_j)$. Also, $v_i(\Hat{B_j}\setminus\{g\})<1$, since otherwise, $g$ would not have been added and $i$ would have claimed $\Hat B_j$. Hence we have 
    \[v_i(\hat B_j\setminus\{g_j\})\leq v_i(\Hat B_j\setminus\{g\})<1.\]
    Adding up the inequalities gives
    \[\sum_{j=1}^nv_i(\Hat B_j\setminus\{g_j\})<\sum_{j=1}^n\max\{1,v_i(j)\}.\]
    Since $m>2n$, at least one of the contributing inequalities is of the second kind, and so is strict. This yields the assertion.
\end{proof}

\begin{lemma}\label{3/2ordered:lem:reduction}
    For all $k\in[\ceil{\frac{3n}{2}}]$, it holds
    \[\MMS_i^k(M\setminus[k+1,3n-k])\geq 1.\]
    In particular, $\MMS_i^n(M\setminus[n+1,2n])\geq 1$. 
\end{lemma}

\begin{proof} 
    Notice that for $k=\ceil{\frac{3n}{2}}$, we have $\MMS^k_i(M)=1$ by assumption. We show by induction, that by decreasing $k$ one at a time, the left hand side never decreases. Consider a $\MMS^k(M\setminus[k+1,3n-k])$ partition for agent $i$. If $k\neq 3n-k+1$ (which is only violated for $k=\ceil{3n/2}$ and $n$ odd), by the pigeonhole principle, two of the items $[k]\cup\{3n-k+1\}$ are in the same bag. We swap these two items with $\{k,3n-k+1\}$, and redistribute the rest of the goods in this bag arbitrarily. If $k=3n-k+1$, we simply redistribute the other goods in the bag containing $k$. We have now created a partition into $k$ bags, one of which is exactly $\{k,3n-k+1\}$, and each of the others having value no less than before. Hence, by removing this bag we have a $(k-1)$ partition where each bag has value $\geq 1$ to agent $i$, hence $\MMS_i^{k-1}(M \setminus [k,3n-k+1])\geq 1$.
\end{proof}

\begin{lemma}\label{3/2ordered:lem:biggoods}
    Let $G=\{g\in M\mid v_i(g)>1\}$, i.e. the items that agent $i$ values more than one. 
    We have $v_i(M\setminus[n+1,2n])\geq n+\sum_{g\in G}(v_i(g)-1)$.
\end{lemma}

\begin{proof}
     In a $\MMS^n(M\setminus[n+1,2n])$ partition $P$ for agent $i$, each bag has value $\geq 1$, by Lemma \ref{3/2ordered:lem:reduction}. If a given bag $B$ in $P$ contains goods from $G$, its value is 
     $$v_i(B)\geq\sum_{g\in G \cap B}v_i(g)\geq1+ \sum_{g\in G\cap B}(v_i(g)-1).$$
     Hence, we have $$v_i(M\setminus[n+1,2n])=\sum_{B\in P}v_i(B)\geq n+\sum_{g\in G}(v_i(g)-1).$$
\end{proof}

\begin{lemma}\label{3/2ordered:lem:MMS}
    In Algorithm \ref{3/2ordered:alg:main}, all agents receive a bundle worth at least 1.
\end{lemma}
\begin{proof}
    Let $G=\{g\in M\mid v_i(g)>1\}$, and $g_j$ be the second most valuable good in $\Hat B_j$ if it exists (and $\emptyset$ otherwise) as before. Additionally, let $\bar M=M\setminus\{g_j:j\in [n]\}$. Notice that $v_i(\{g_j:j\in[n]\})\leq v_i(\{n+1,\dots,2n\})$. We then have
    \begin{align*}
        v_i(\bar M)&\geq v_i(M)-v_i(\{n+1,\dots,2n\})\\
        &\geq n+\sum_{g\in G}(v_i(g)-1)  &\text{(Lemma  \ref{3/2ordered:lem:biggoods})}\\
        &=\sum_{j=1}^n\max\{1,v_i(j)\}\\
        &>\sum_{j=1}^n v_i(\Hat B_j\setminus\{g_j\}) & \text{(Lemma \ref{3/2ordered:lem:leftover})} \\
        &=v_i(\bar M);
    \end{align*}
    a contradiction. Hence, every agent receives a bag.
\end{proof}

\begin{lemma}\label{3/2ordered:lem:EFX}
    The allocation computed by Algorithm \ref{3/2ordered:alg:main} is EFX.
\end{lemma}
\begin{proof}
    Assume an agent $i$, having received a bag $B$ strongly envies an agent $j$, holding bag $C$, in the end. Notice that the items in $C$ were added in decreasing order of value, so in particular, the last item added is the least valuable. Hence we only need to consider this good for showing EFX. We now distinguish 4 cases:
    
    \textbf{Case 1:} $C=\{g\}$, i.e. a singleton claimed during setup. Singleton Bags can never be strongly envied.
    
    \textbf{Case 2:} $C=B_k$, i.e. a two item bag claimed right after setup. Now $v_i(k)\geq 1$, so $i$ should have claimed a singleton during setup.

    \textbf{Case 3:} $C$ has received items during bag filling, and $i$ had not claimed a bag at that time. Let $g$ be the last good added to $C$. Now $v_i(C\setminus\{g\})<1\leq v_i(B)$, otherwise $g$ would not have been added. Since $g$ is the least valuable good in $C$, there is no strong envy.

    \textbf{Case 4:} $C$ has received items during bag filling, and $i$ had a claimed bag $B'$ at that time. Let $g$ be the last item added to $C$. Now $v_i(C\setminus\{g\})<v_i(B')$, since otherwise $i$ would have swapped for $v_i(C\setminus\{g\})$, instead of adding $g$. Also $v_i(B')\leq v_i(B)$, since the value of each agents bag only increases by swapping bags. Again, since $g$ is the least valuable good in $C$ there is no strong envy.
\end{proof}

Combining these results, and then applying the envy-cycle elimination algorithm, we get:
\theoremOne*
\begin{proof}
    By Lemmas \ref{3/2ordered:lem:MMS} and \ref{3/2ordered:lem:EFX}, Algorithm \ref{3/2ordered:alg:main} computes a partial allocation that is 1-out-of-$\ceil{3n/2}$ MMS and EFX. Since the goods were added in decreasing order of value, the requirements for Lemma \ref{3/2ordered:lem:envy_cycles} are met. We can thus apply the envy-cycle elimination algorithm to the result of Algorithm \ref{3/2ordered:alg:main}, and we obtain a complete allocation that does not decrease the share for any agent.
\end{proof}

\section{\boldmath $1$-out-of-$\ceil{3n/2}$ MMS Together with EF1 for Top-$n$ Instances}\label{sec4}

In this section, we present an algorithm that gets a top-$n$ instance and outputs a partial allocation that is $1$-out-of-$\ceil{3n/2}$ MMS and EFX at the same time (Algorithm \ref{3/2top:alg:lone_divider}), and a way to complete this allocation while maintaining EF1. Algorithm \ref{3/2top:alg:lone_divider} is algorithmically similar to the one presented in \cite{Hosseini2021OrdinalMS}, but with differences in important details, and differs majorly in its analysis. Firstly, when wanting only the MMS guarantee, our analysis is significantly shorter and simpler. Also, since we do not have orderedness, to achieve an ordinal MMS guarantee, we maintain the additional invariant that every bag created throughout the algorithm contains exactly one of the top $n$ goods of the instance.

\begin{algorithm}[tb]
    \DontPrintSemicolon

    \SetKwInOut{Input}{Input}
    \SetKwInOut{Output}{Output}
    \Input{A top-$n$ instance $(M,N,(v_i)_{i\in N})$}
    \Output{A partial allocation that is 1-out-of-$\ceil{3n/2}$ MMS and EFX}
    \BlankLine
    $B_j\gets \emptyset$ for all $j\in N$, \quad $U\gets N$, \quad $L\gets M$\;
    \While{$U\neq \emptyset$\label{ln:lone_divider:while}}{
        Let  an arbitrary ``lone divider'' agent $i\in U$ divide the items in $L$ into $|U|$ bags $C_j$, each with value $\geq 1$ to $i$ and each containing exactly one of the top $n$ goods\label{ln:lone_divider:lone_divider} \;
        \For{$j\in U$}{
            Let $C'_j \subseteq C_j$ be inclusion-wise minimal, containing its top $n$ good, such that $\exists i'\in U$ with $v_i(C_j')\geq 1$\label{ln:lone_divider:shrink}\;
            \If{$\exists a\in N\setminus U$ such that $a$ strongly envies $C_j$\label{ln:lone_divider:if}}{
                Let $a^*\in N\setminus U$ be a most envious agent of $C_j$\; 
                Let $C_j''\subseteq C_j'$ be inclusion-wise minimal, containing its top $n$ good, with $v_{a^*}(C_j'')>v_{a^*}(B_{a^*})$ for some agent $a^*$, and no agent strongly envies $C_j''$ \;
                $L \gets L\cup B_{a^*}\setminus C_j'',\quad B_{a^*} \gets C_j''$ \label{ln:lone_divider:swap}\;
                Go to Line \ref{ln:lone_divider:while}\;
            }
        }
        Find a nonempty envy-free matching in the threshold graph $T_{C',\bone}$ between the bags $C'=(C_1',\dots,C_{|U|}')$ and $U$\label{ln:lone_divider:ef_matching}\;
        \For{Agents $i$ matched to a bag $C_j'$}{
            $B_i \gets C_j'$, \quad $U \gets U\setminus\{i\}$, \quad $L \gets L\setminus B_i$\label{ln:lone_divider:assign}\;
        }
    }
    \Return $(B_1,\dots ,B_n)$\;
    \caption{\sc{LoneDivider}}\label{3/2top:alg:lone_divider}
\end{algorithm}

In order to state the algorithm, we require the following two definitions:

\begin{definition}
    Given a set of agents $U$ and a partition of (a subset of) the goods $B=(B_i)_{i\in U}$, as well as a threshold vector $t=(t_i)_{i\in U}$, the \emph{threshold graph $T_{B,t}$} is a bipartite graph with nodes corresponding to agents in $U$ on one side and nodes corresponding to bags in $B$ on the other side, with an edge between an agent $i$ and a bag $B_j$ iff $v_i(B_j)\geq t_i$.
\end{definition}

\begin{definition}
    Given a threshold graph $T_{B,t}$, a matching $M$ in $G$ is called \emph{envy-free}, if for every matched bag $B_j\in B$ and unmatched agent $i$, we have $v_i(B_j)<t_i$, i.e. there is no edge from $i$ to $B_j$.
\end{definition}

The algorithm starts by assigning an arbitrary agent $i$ the role of ``lone divider''. They are tasked to divide the goods into bags that $i$ values at least 1, while disregarding other agents' valuations. Next, we compute an envy-free matching. This serves two purposes: it maintains envy-freeness from unmatched agents to the matched ones, while also ensuring that leftover goods still have sufficient value to satisfy the MMS guarantee for the remaining agents. Now we shrink some bags in order to eliminate envy between matched agents. Lastly, we give agents who were assigned a bundle in a previous iteration the option to swap for a bag we have just created. This ensures that envy cannot be created between agents who were assigned in different iterations. If this was not necessary, we assign each agent the bag they were matched to, and iterate this process until every agent receives a bag.

Algorithm \ref{3/2top:alg:lone_divider} relies on finding envy-free matchings in the threshold graph with threshold 1 for all agents, since every agent is meant to receive bag of value at least 1. 

\begin{restatable}[\cite{AkramiRathi2025simultaneous}]{lemma}{lemEFMatching}\label{3/2top:lem:EFMatching}
    The threshold graph $T_{C',\bone}$ from Algorithm \ref{3/2top:alg:lone_divider} always contains a nonempty envy-free matching.
\end{restatable}
For completeness, we provide the proof of Lemma \ref{3/2top:lem:EFMatching} in Appendix \ref{app4}.

The main thing left to show is that Line \ref{ln:lone_divider:lone_divider} can always be executed, i.e. any given agent $i$ can be a lone divider, as well as that the algorithm terminates. 


\begin{lemma}\label{3/2top:lem:lone_divider}
    In line \ref{ln:lone_divider:lone_divider} of Algorithm \ref{3/2top:alg:lone_divider}, any agent $i \in U$ can fulfill the role of lone divider, i.e. can divide the unallocated goods $L$ into $|U|$ bundles of value at least 1, each of which contains exactly one of the top $n$ goods.
\end{lemma}
\begin{proof}
    Let $S$ be the set of agents who are already satisfied at the start of the iteration where $i$ is meant to become a lone divider, let $k=|S|$, and let $(B_j)_{j\in S}$ be the bags of these agents. Also, without loss of generality, we can order the goods in $M$ by non-increasing $v_i$. In particular, the top $n$ goods are now exactly $[n]$. First observe that $v_i(B_j)<1$ for all $j\in S$: all matchings up to this point have been envy-free, so all bags allocated in Line \ref{ln:lone_divider:assign} have value less than 1. 
    
    Bags assigned in Line \ref{ln:lone_divider:swap} are proper subsets of $C'_j$, as some agent $a$ strongly envied $C'_j$. Now since $C'_j$ was minimal, $v_i(C_j'')<1$. We also know that exactly $n-k$ of the top $n$ goods are still in $L$, since every bag created throughout the algorithm contains exactly one of them. We now find a partition of $L$ by using Algorithm \ref{3/2ordered:alg:main}: we apply the algorithm to $L$ and $n-k$ copies of agent $i$. Since all agents are identical, this is now an ordered instance with respect to the order we chose. Note in particular that the most valuable good $g_j'$ in each bag $B_j, j\notin S$ is from $[n]$, and since there are exactly $n-k$ many of those left, the second most valuable good $g_j$ in each bag is from $M\setminus[n]$ (if it exists). Let $\Bar L = L\setminus\{g_j\mid j\notin S\}$, i.e. all leftover items, except the second most valuable in each bag. 
    Define $G=\{g\in L\mid v_i(g)>1\}$. We now assume that one of the bags at the end of the algorithm has value $<1$ to $i$. This is equivalent to one of the copies of $i$ not receiving a bag. We do however need to modify the proof of Theorem \ref{3/2ordered:lem:MMS} slightly to account for the missing agents and goods:
    \begin{align*}
        v_i(\Bar L)&=v_i(M)-v_i(M\setminus L)-\sum_{j\notin S} v_i(g_j)\\
        &\geq v_i(M)-v_i(M\setminus L)-v_i([n+1,2n])\\
        &>v_i(M)-k-v_i([n+1,2n])\\
        &\geq n-k+\sum_{g\in G}(v_i(g)-1)&\hidewidth\text{(Lemma \ref{3/2ordered:lem:biggoods})}\\
        &\geq \sum_{j\notin S} \max\{1,v_i(g_j')\} > v_i(\Bar L). &\hidewidth\text{(Lemma \ref{3/2ordered:lem:leftover})}
    \end{align*}
    This is, again, a contradiction. Hence, $i$ can divide the goods in the way that line \ref{ln:lone_divider:lone_divider} of Algorithm \ref{3/2top:alg:lone_divider} requires.
\end{proof}

The proof of the following three lemmas are similar to the ones provided in \cite{AkramiRathi2025simultaneous}.

\begin{lemma}\label{3/2top:lem:LDterminates}
    Algorithm \ref{3/2top:alg:lone_divider} terminates.
\end{lemma}
\begin{proof}
    If the condition in Line \ref{ln:lone_divider:if} is true, $a^*$ receives a bag worth more than before, while everything else stays the same, so the total value of allocated goods increases. If it is false, the algorithm allocates at least one bag to an unallocated agent by Lemma \ref{3/2top:lem:EFMatching}.

    Now $(\sum_{i\in N}v_i(B_i), |N\setminus U|)$ increases in every iteration and the algorithm terminates.
\end{proof}

\begin{lemma}\label{3/2top:lem:LDMMS}
    Algorithm \ref{3/2top:alg:lone_divider} returns a partial allocation that is 1-out-of-$\ceil{3n/2}$ MMS.
\end{lemma}
\begin{proof}
    Agents are only assigned bags worth at least 1, since that is the threshold value in Line \ref{ln:lone_divider:ef_matching}. In the swap step of Line \ref{ln:lone_divider:swap} the value of the bag assigned to $a^*$ only increases.
\end{proof}

\begin{lemma}\label{3/2top:lem:LDEFX}
    Algorithm \ref{3/2top:alg:lone_divider} returns a partial allocation that is EFX.
\end{lemma}
\begin{proof}
    We prove this claim via induction, i.e. that at any given point in the algorithm, the allocation restricted to $N\setminus U$ is EFX, and for $i\in U$, $v_i(B_j)<1$ for all assigned bags $B_j$. At the beginning, the allocation is empty, so the claim holds. The two steps where this can change are the swap in Line \ref{ln:lone_divider:swap} and the assignment in Line \ref{ln:lone_divider:assign}.\\

    In the swap operation, by choice of $C_j''$, none of the assigned agents may now strongly envy $a^*$. The agent $a^*$ can also not envy another agent, since their value only increased. Also, since $C_j'$ was already minimal and $C_j''\subsetneq C_j'$, we have $v_i(C_j'')<1$ for all $i\in U$.\\

    In the assign operation we need to check envy from three groups towards the newly assigned bags: The agents who are still unassigned value each new bag less than 1, since the matching was envy-free. The newly assigned agents cannot strongly envy one another, since we chose the $C_j'$ in line \ref{ln:lone_divider:shrink} minimal; any strong envy towards a bag $C$ implies that $C$ could have been shrunken further. And finally, no already assigned agents may strongly envy a new bag, since the condition in Line \ref{ln:lone_divider:if} was false. We also check envy of the newly assigned agents $i$ towards already satisfied agents $i'$. By induction, $v_i(B_{i'})<1\leq v_i(B_i)$, so this case also does not produce any envy.

    Hence, the algorithm always maintains a partial allocation that in EFX.
\end{proof}

\theoremTwo*

\begin{proof} 
    Lemmas \ref{3/2top:lem:LDterminates}, \ref{3/2top:lem:LDMMS} and \ref{3/2top:lem:LDEFX} together imply that Algorithm \ref{3/2top:alg:lone_divider} computes a partial allocation that is both 1-out-of-$\ceil{3n/2}$ MMS and EFX.
    
    To see the second assertion, consider the envy cycle elimination procedure from \cite{lipton2004approximately} again, which we discussed in Lemma \ref{3/2ordered:lem:envy_cycles}. Since the instance is not ordered, we can only use the weaker formulation. So applying this procedure to the result of Algorithm \ref{3/2top:alg:lone_divider}, we get a full allocation that is both 1-out-of-$\ceil{3n/2}$ MMS and EF1.
\end{proof}

\section{\boldmath $1$-out-of-$4\ceil{n/3}$ MMS Together with EF1 for Ordered Instances}\label{sec5}
In this section we present an algorithm that computes an allocation that is both 1-out-of-$4\ceil{n/3}$ MMS and EF1 for ordered instances (Algorithm \ref{4/3ordered:alg:main}). It is conceptually similar to the one presented in \cite{MMS-l-out-of-d}, with significant adjustments that make it compatible with envy based notions. See Appendix \ref{app5} for the missing proofs of this section.

In this section, we assume $n$ to be a multiple of 3. We can achieve this by copying an arbitrary agent 1 or 2 times. This does not change $4\ceil{n/3}$, and does not disrupt the ordering condition. Hence, a complete 1-out-of-$4\ceil{n/3}$ MMS and EF1 allocation for this modified instance is a partial allocation with the same guarantees for the original instance. If we wish to obtain a complete allocation of the goods, we need to decide whom to allocate the goods that were previously allocated to the dummy agents. We do this in Lemma \ref{4/3ordered:lem:EF}.

\subsection{Working with Normalization}
Many state of the art algorithms for computing share based allocations use normalization in their analysis:
\begin{definition}\label{def:normal}
    An instance $(M,N,(v_i)_{i\in N})$ of fair division is called \emph{$d$-normalized}, if in a $\MMS^d$ partition for any agent $i$, all bundles have value exactly 1.
\end{definition}
This can be assumed without loss of generality if we only want share based fairness guarantees.
\begin{restatable}[\cite{simple}]{lemma}{lemNormalShare}\label{4/3ordered:lem:normal_share}
    If a share-based fairness result holds for $d$-normalized instances, it holds for all instances with $\MMS^d_i>0$ for all agents $i$.
\end{restatable}

A common way of normalizing an instance is to do the following. Let $P=(P_1, \ldots, P_d)$ be a $\mu^d(M)$ partition for agent $i$. Then for all $g \in M$ and $j \in [d]$, if $g \in P_j$, we set $\bar{v}_i(g) = v_i(g)/v_i(P_j)$. 
If we also want to include an envy based guarantee, normalization can no longer be generated this way.
\begin{example}\label{example}
    We look at the following instance $I=(M,N,(v_i)_{i\in N})$ and a normalization $I'$: $M=[5],N=[3]$, $v_1=v_2\equiv 1$ and $v_3(g)=1$ for $g\in\{1,2,3,5\}, v_3(4)=2$. Now $(\{1,2\},\{3,4\},\{5\})$ is a $\MMS^3$ partition for agents 1 and 2, hence $\Normal v_1(g)=\Normal v_2(g)=1/2$ for $g\in\{1,\dots,4\}$, and  $\Normal v_1(5)=1$. On the other hand, $(\{1,2\},\{3,5\},\{4\})$ is a $\MMS^3_3$ partition. Hence $\Normal v_3(g)=1/2$ for $g\in\{1,2,3,5\}$ and $\Normal v_3(4)=1$. The allocation $(\{5\},\{1,2,3\},\{4\})$ is therefore EFX in $I'$, but in $I$ agent 1 strongly envies agent 2.
\end{example}
In order to adapt algorithms that utilize normalized instances it is still instructive to look into methods for normalization. In particular, since we can not order the input instance ourselves, and instead assume it to be ordered, we have to ensure that the method we use to normalize the instance respects the given order; we can not, as is done in the usual analysis, normalize first, and then order the result. 

\begin{lemma}\label{4/3ordered:lem:normalization}
    Given an ordered instance $I=(M,N,(v_i)_{i\in N})$ of fair division, there is an instance $I'=(M,N,(\Normal v_i)_{i\in N})$ that is $d$-normalized and ordered with respect to the order of $I$, and has $\Normal v_i(g)\leq v_i(g)$ for all $g\in M$.
\end{lemma}
\begin{proof}
    We construct a method of normalizing $I$ that respects its ordering. Let $P=(P_1,\dots,P_d)$ be an $\MMS^d$ partition for an agent $i$ and $P_j\in P$ be a bundle with $v_i(P_j)>1$. We continuously and uniformly shrink goods in $P_j$ for agent $i$ until either $\Normal v_i(P_j)=1$, in which case we stop and move on to the next bundle, or a good $g\in P_j$ has equal value to any good $h_k\notin P_j$. In that case, we swap $g$ with the good $h$ that has the highest index among the candidates for $h_k$. We then continue shrinking.

    This method terminates, since swaps only occur at times where goods have values that other goods already had. These are just the beginning valuations and any values that goods in already shrunken bags have. This can only happen finitely many times. It also computes a $d$-normalized instance, as in the end we have a $\MMS^d$ partition where every bundle has value 1. The ordering remains intact by construction. 
\end{proof}

\begin{remark}
    Since this method does not increase the value of goods, the argument from Lemma \ref{4/3ordered:lem:normal_share} still works: For any share based fairness guarantees we can without loss of generality assume that the instance has been normalized in this way.
\end{remark}

We have now established that we can assume normalization that respects the given ordering of an instance if we wish to attain an MMS guarantee of some sort, but not if we simultaneously wish for an envy-based guarantee. On the other hand, analysis of the algorithm in Section \ref{sec:4/3ordered:algorithm} requires normalization in its MMS guarantee. To circumvent this, we run Algorithm \ref{4/3ordered:alg:main} using the original valuations $v_i$. During the analysis of the MMS guarantee, we will then consider the normalized valuations $\Normal v_i$, as given by Lemma \ref{4/3ordered:lem:normalization}. Since normalizing in this way cannot increase valuations, if we can prove that every agent receives their $\MMS^d$ value with respect to $\Normal v$, they will also receive it with respect to $v$. On the other hand, whenever we consider envy throughout the algorithm, it will be with respect to the original valuations $v$, such that normalization can not interfere with our envy-freeness guarantee.

Some notable points to consider for this concept are, that all arguments in the analysis that call upon the ordering can do so before and after normalization, since Lemma \ref{4/3ordered:lem:normalization} guarantees that the order does not change. In particular, if the algorithm handles goods in a particular order, this property is available in the analysis. Furthermore, some bags $B$ might have $v_i(B)\geq1>\Normal v_i(B)$, so the algorithm considers $B$ to be fit for agent $i$ while the analysis does not. Hence we can never consider the valuations of satisfied agents during the MMS part of the analysis. We may only upper bound the value of assigned bags for 
unsatisfied agents, which works since normalization can only decrease values of goods. 

\subsection{The Algorithm}\label{sec:4/3ordered:algorithm}

\begin{algorithm}[tb]
    \DontPrintSemicolon

    \SetKwInOut{Input}{Input}
    \SetKwInOut{Output}{Output}
    \Input{An ordered instance $(M,N,(v_i)_{i\in N})$}
    \Output{A partial allocation that is 1-out-of-$4\ceil{n/3}$ MMS and EF1}
    \BlankLine
    \For(\tcp*[f]{Initialization}){$j$ from $1$ to $n$}{
        $B_j\gets\{j,2n-j+1\}$ \;
    }
    $g\gets2n+1$,\quad $L\gets\{B_1,\dots,B_n\}$, \quad $U\gets N$\;
    \While(\tcp*[f]{Bag-Filling}){$U\neq \emptyset$}{
        \uIf{$\exists i\in U, B\in L:v_i(B)\geq 1$}{
            $L\gets L\setminus\{B\}$,\quad $U\gets U\setminus\{i\}$\;
            $A_i\gets B$\;
        }\uElseIf{$\exists i\in N\setminus U, B\in L:v_i(B)>v_i(A_i)$\label{ln:4/3ordered:swap}}{
            $L\gets (L\cup\{A_i\})\setminus \{B\}$\;
            $A_i\gets B$\;
        }\Else{
            Let $B\in L$ arbitrary: $B\gets B\cup\{g\}$\;
            $g\gets g+1$\;
        }
    }
    \Return $(A_1,\dots ,A_n)$\;
    \caption{\sc{1-out-of-$4\ceil{n/3}$ MMS + EF1}}
    \label{4/3ordered:alg:main}
\end{algorithm}

Algorithm \ref{4/3ordered:alg:main} is a bag-filling procedure 
 and differs from the one presented in \cite{MMS-l-out-of-d} mainly in Line \ref{ln:4/3ordered:swap}. This is a swap step similar to the one seen in Section \ref{sec:3/2ordered}. We have seen there that this step does not disturb the MMS guarantees, if the analysis is sufficiently symmetrical, i.e. does not rely 
on the order of the bags, items and agents. This is not the case for the analysis presented in \cite{MMS-l-out-of-d}, as it uses that bags are filled one after the other, and by increasing index of the bags. This is a property that is not maintainable throughout the swap step. It nonetheless prompts us to consider that analysis, and adapt it in a way that renders this condition redundant. 

To show that the allocation produced by Algorithm \ref{4/3ordered:alg:main} is in fact 1-out-of-$4\ceil{n/3}$ MMS, we begin by assuming towards a contradiction that some agent $i$ does not receive a bag. This happens if the algorithm runs out of goods to fill bags before agent $i$ is assigned a bag. In particular, there exists an initial bag $B_j$ with $\Normal v_i(B_j)\leq v_i(B_j)<1$. Let $\ell^*$ be the lowest index such that $\Normal v_i(B_{\ell^*+1})<1$. We distinguish the cases $\Normal v_i(2n-\ell^*)\geq 1/3$ and $\Normal v_i(2n-\ell^*)<1/3$.

The first of these cases uses a partition of the bags into three groups. In \cite{MMS-l-out-of-d}, this partition is dependent on the algorithm filling the bags one after the other, and in increasing order of index. Since the swap step in Algorithm \ref{4/3ordered:alg:main} makes this infeasible, we resort to Definition \ref{4/3ordered:def:partition}. That these groups behave functionally identical becomes clear in the detailed proofs provided in Appendix \ref{app5}. 

\begin{definition}\label{4/3ordered:def:partition}
    Let $A^+=\{B_1,\dots,B_{\ell^*}\}$. Let $A^1$ be the set of all bags that receive exactly one item during bag filling before any bag receives two items. Let $A^2$ be the set of all other bags. 
\end{definition}

Notice in particular that $A^+\cap A^1=\emptyset$, since before an item would be added to $B_j\in A^+$, agent $i$ could claim $B_j$.

With minor adjustments to the analysis from \cite{MMS-l-out-of-d}, concerning the normalization and the aforementioned Definition \ref{4/3ordered:def:partition}, we achieve the following theorem:

\begin{restatable}{theorem}{PartOne}\label{4/3ordered:thm:part1}
    If Algorithm \ref{4/3ordered:alg:main} does not allocate a bag to some agent $i$, then $\Normal v_i(2n-\ell^*)<1/3$.
\end{restatable}

For the second case, we only bound the value of the initial bags $B_j$. Hence, neither the normalization, nor the deviation from the order of the bags during bag-filling can impact the analysis. All proofs work precisely as they do in \cite{MMS-l-out-of-d}. We thus achieve:

\begin{theorem}[\protect{\cite{MMS-l-out-of-d}}]\label{4/3ordered:thm:part2}
    If Algorithm \ref{4/3ordered:alg:main} does not allocate a bag to some agent $i$, then $\Normal v_i(2n-\ell^*)\geq 1/3$.
\end{theorem}

Theorems \ref{4/3ordered:thm:part1} and \ref{4/3ordered:thm:part2} together cover all cases for $\Normal v_i(2n-\ell^*)$ and lead to a contradiction. Hence, no agent can be left without a bag at the end of Algorithm \ref{4/3ordered:alg:main}.

The last thing towards proving the correctness of Algorithm \ref{4/3ordered:alg:main} is to show that the partial allocation is in fact EF1.

\begin{lemma}\label{4/3ordered:lem:EF}
    The allocation computed by Algorithm \ref{4/3ordered:alg:main} is EF1.
\end{lemma}
\begin{proof}
    Let $i$ and $j$ be two agents, having been assigned bags $A_i, A_j$ respectively. We show that there exits $g \in A_j$ such that $v_i(A_i) \geq v_i(A_j \setminus \{g\})$. 
    We distinguish three cases:

    \textbf{Case 1:} $A_j=B_k=\{k,2n-k+1\}$, i.e. a bag created during initialization. We can choose to remove good $k$. Now $A_i$ contains a good from $[n]$, which is more valuable than $2n-k+1$, so $v_i(A_i)\geq v_i(A_j\setminus\{k\})$.

    \textbf{Case 2:} $A_j$ has received goods during bag-filling, and $j$ was assigned $A_j$ before $i$ was assigned a bag. In this case, let $g$ be the last good added to $A_j$. Now $v_i(A_j\setminus\{g\})<1\leq v_i(A_i)$, since otherwise, $i$ would have claimed $A_j\setminus\{g\}$ (before adding $g$).

    \textbf{Case 3:} $A_j$ has received goods during bag-filling, and $j$ was assigned $A_j$ while $i$ was already assigned a bag $A$ at that time. Let $g$ be the last good added to $A_j$. Now $v_i(A_j\setminus\{g\})\leq v_i(A)$, since else $i$ would have swapped bags instead of $g$ being added. Also, $v_i(A)\leq v_i(A_i)$, since agents only swap bags for higher value ones. Thus, $v_i(A_j\setminus\{g\})\leq v_i(A_i)$. 
\end{proof}

\theoremThree*
\begin{proof}
    Theorems \ref{4/3ordered:thm:part1} and \ref{4/3ordered:thm:part2} together with Lemma \ref{4/3ordered:lem:EF} show that the partial allocation resulting from Algorithm \ref{4/3ordered:alg:main} is 1-out-of-$4\ceil{n/3}$ MMS and EF1. In order to complete this allocation, we turn back to Lemma \ref{3/2ordered:lem:envy_cycles}. Since we do not only need to allocate the goods that were possibly leftover after Algorithm \ref{4/3ordered:alg:main} concludes, but also the goods that were allocated to dummy agents, we cannot use the stronger version that promises EFX. Instead, we can only maintain the EF1 guarantee by disregarding the ordering.
    
    In conclusion, inputting the result from Algorithm \ref{4/3ordered:alg:main} into the envy-cycle elimination procedure from Lemma \ref{3/2ordered:lem:envy_cycles} yields a complete allocation that is 1-out-of-$4\ceil{n/3}$ MMS and EF1 at the same time.
\end{proof}

\section{Conclusion}

We studied the compatibility between \emph{ordinal} maximin share guarantees and envy-based fairness notions in the allocation of indivisible goods. Focusing on ordered and top-\(n\) instances, we showed that strong ordinal MMS guarantees can be achieved simultaneously with EFX and EF1. 

Our results highlight the potential of ordinal approaches for reconciling share-based and envy-based fairness. An important direction for future work is to understand whether similar guarantees can be obtained for general instances or with stronger envy-based notions, and to further explore the limits of combining ordinal and envy-based fairness criteria.

\bibliographystyle{alpha}
\bibliography{ref}

\newpage 
\appendix
\newpage

\section{Missing Proofs of Section \ref{sec2}}\label{app2}
\lemEnvyCycle*
\begin{proof}
    We analyze Algorithm \ref{common:alg:envycycle}. It works by adding remaining goods to bags in $A$, while maintaining EF1 or EFX respectively. Consider the digraph $G$ where the nodes are agents, and there exists an edge $(i,j)$ if $i$ envies $j$. If $G$ contains an agent $i$ with no incoming edge, we can add the most valuable unassigned item $g$ to $i$'s bundle. If we are in the weaker case of the lemma, we maintain EF1: no agent envied $i$ before, so removing $g$ from $A_i$ resolves all envy. Otherwise, since $I$ is ordered, and $A$ consisted of the most valuable goods, $g$ is now $i$'s least valuable good for every agent. Removing $g$ from $A_i$ resolves all envy as before, and since it is the least valuable good, so does removing any other good.

    If $G$ contains no such agent, then it must contain a directed cycle $C$. In this case, we can assign each agent in $C$ the bundle of the next agent along $C$. This does not change the bundles, only their owner, and each affected agent gets a bag of higher value. Thus, the result is still EF1 or EFX respectively. Now the total utility $\sum_{i\in N}v_i(A_i)$ has strictly increased, so we only need to eliminate finitely many cycles this way to find an agent with in-degree 0 to assign a new good to.

    Both of these steps keep the requirements of the lemma invariant, so can be iterated until all goods are assigned.
\end{proof}

\begin{algorithm}[tb]
    \DontPrintSemicolon

    \SetKwInOut{Input}{Input}
    \SetKwInOut{Output}{Output}
    \Input{A partial allocation $A$}
    \Output{A complete allocation}
    \BlankLine
    \While{$M\setminus A\neq \emptyset$}{
        \While{The envy graph $G$ does not contain a source}{
            Let $C$ be a cycle in $G$\;
            $A_i\gets A_j$ for $(i,j)\in C$\;
        }
        Let $i$ be a source in $G$, let $g\in M\setminus A$ with $v_i(g)$ maximum\;
        $A_i\gets A_i\cup\{g\}$\;
    }
    \Return $(A_1,\dots ,A_n)$\;
    \caption{\sc{EnvyCycleElimination}}
    \label{common:alg:envycycle}
\end{algorithm}


\section{Missing Proofs of Section \ref{sec4}}\label{app4}
\lemEFMatching*
\begin{proof}
    A perfect matching is envy free, so if $T_{C',\bone}$ contains a perfect matching, the lemma holds. If not, there exists an inclusion wise minimal Hall-violating set of bags $X\subseteq C'$, i.e. $|N(X)|<|X|$, and this condition holds for no $Y\subsetneq X$. Hence, for every $\emptyset\neq Y\subsetneq X$ there exists a perfect matching between $Y$ and $N(Y)$. This matching is envy-free, since any agent who envies a bag in $Y$ is already in $N(Y)$ and thus matched. This proper subset $Y$ exists since every bag has at least one incident edge in $T_{C',\bone}$ by construction of the shrink step in Line \ref{ln:lone_divider:shrink}; This means $|X|>1$.
\end{proof}

\section{Missing Proofs of Section \ref{sec5}}\label{app5}

\lemNormalShare*
\begin{proof}
    Let $I=(M,N,(v_i)_{i\in N})$ be an instance of fair division. Let $P=(P_1,\dots,P_d)$ be a $\MMS^d$ partition for an agent $i$. Let $P_{min}=\min_Pv_i(P_j)$. Now setting $\Normal v_i(g)=v_i(P_{min})v_i(g)/v_i(P_j)$ with $g\in P_j$ for all $g\in M, i\in N$ induces a new instance $I'=(M,N,(\Normal v_i)_{i\in N})$ that is normalized, and $\Normal v_i(g)\leq v_i(g)$. Now by assumption, $I'$ can be fairly divided. We now take a solution of $I'$ as our solution for $I$. Each agent receives at least as much value as they received with respect to $\Normal v_i$, so the solution remains valid.
\end{proof}

For completeness, we now give part of the proof of correctness for algorithm \ref{4/3ordered:alg:main}. In particular, we give the modified proofs for 1-out-of-$4\ceil{n/3}$ MMS in the case $\Normal v_i(2n-\ell^*)\geq 1/3$. In all of the following, we will refer to the bag initialized as $\{j,2n-j+1\}$ as $B_j$ and the bag resulting from $B_j$ after the algorithm as $\Hat B_j$. Let $\Normal v_i(2n-\ell^*)=1/3-x$ for some $x>0$.

First, let us make two quick observations, which have similar proofs to \cite{MMS-l-out-of-d}:

\begin{lemma}\label{4/3ordered:lem:obs1}
    For all $j\leq k\leq n$, we have $\Normal v_i(\Hat B_j)\leq 1+\Normal v_i(2n-k+1)$.
\end{lemma}
\begin{proof}
    Let $g$ be the least valuable good in $\Hat B_j$. If $g=2n-j+1$, we use that $\Normal v_i$ is normalized, so $\Normal v_i(j)\leq 1$. Hence
    \[\Normal v_i(\Hat B_j)=\Normal v_i(\{j,2n-j+1\})\leq 1+\Normal v_i(\{2n-k+1\}).\]
    Otherwise, $g$ was added during the bag-filling phase, meaning that $\Normal v_i(\Hat B_j\setminus\{g\})\leq v_i(\Hat B_j\setminus\{g\})<1$, since else $i$ would have been assigned the bag before adding $g$. Now
    \[\Normal v_i(\Hat B_j)=\Normal v_i(\Hat B_j\setminus\{g\})+\Normal v_i(g)\leq 1+\Normal v_i(2n-k+1).\]
\end{proof}

\begin{lemma}\label{4/3ordered:lem:obs2}
    For all $k\leq j\leq n$, we have $\Normal v_i(\Hat B_j)\leq \max(1+\Normal v_i(2n-k+1),2\Normal v_i(k))$.
\end{lemma}
\begin{proof}
    Again, we distinguish two cases. If there were goods added to $B_j$ during bag-filling, let $g$ be the last. Then, like before
    \[\Normal v_i(\Hat B_j)=\Normal v_i(\Hat B_j\setminus\{g\})+\Normal v_i(g)\leq 1+\Normal v_i(2n-k+1).\]
    Otherwise, since $k\leq j$, and hence $\Normal v_i(2n-j+1)\leq \Normal v_i(j)\leq \Normal v_i(k)$, we get
    \[\Normal v_i(\Hat B_j)=\Normal v_i(\{j,2n-j+1\})\leq 2\Normal v_i(k)\]
\end{proof}

\begin{lemma}[\protect{\cite{MMS-l-out-of-d}}]\label{4/3ordered:lem:smallgoods}
    For all $j\geq 2n-\ell^*$, we have $\Normal v_i(j)<1/2$. In particular, $x<1/6$.
\end{lemma}

\begin{definition}   
    Define the following partition of the starting bags:
    First, $A^+=\{B_1,\dots,B_{\ell^*}\}$. Then $A^1$ contains all bags that receive exactly one item during bag filling before any bag receives two items. $A^2$ contains all other bags. Notice in particular that $A^+\cap A^1=\emptyset$, since before an item would be added to $B_j\in A^+$, agent $i$ could claim $B_j$.
\end{definition}

\begin{lemma}\label{4/3ordered:lem:goodbags}
    For all $B_j\in A^2$, we have $\Normal v_i(\Hat B_j)<4/3-2x$.
\end{lemma}
\begin{proof}
    It holds
    \begin{align*}
        1&>\Normal v_i(B_{\ell^*+1})\\
        &=\Normal v_i(\ell^*+1)+\Normal v_i(2n-\ell^*)\\
        &=\Normal v_i(\ell^*+1)+1/3+x,&
    \end{align*}
    so $\Normal v_i(\ell^*+1)<2/3-x$. Now if $\Hat{B_j}=B_j$, we have
    \begin{align*}
        \Normal v_i(\Hat B_j)&=\Normal v_i(B_j)\\&=\Normal v_i(j)+\Normal v_i(2n-j+1)\\
        &\leq 2\Normal v_i(\ell^*+1)&\hidewidth(2n-j+1\geq j\geq \ell^*)\\
        &<4/3-2x.
    \end{align*}
    If $\Hat B_j\neq B_j$, that means that an item was added to $B_j$ during bag filling. Let $h$ be the last item added to $B_j$. Also, let $B_k$ be the first bag to receive two items during bag filling. Then $h$ was added to $B_j$ no earlier than $B_k$ received its first item, otherwise $B_j$ would be in $A^1$. Let $g$ be the first item added to $B_k$. Then, since $B_k$ required another item, $k>\ell^*$, and we get 
    \begin{align*}
        1&>v_i(B_k\cup\{g\})\\
        &\geq\Normal v_i(B_k\cup\{g\})\\
        &=\Normal v_i(k)+\Normal v_i(2n-k+1)+\Normal v_i(g)\\
        &\geq 2\Normal v_i(2n-\ell^*)+\Normal v_i(g)\\
        &=2/3+2x+\Normal v_i(g),
    \end{align*}
    so $\Normal v_i(g)<1/3-2x$. Now $g\leq h$, so $\Normal v_i(g)\geq \Normal v_i(h)$, so we get
    \begin{align*}
        \Normal v_i(\Hat B_j)&=\Normal v_i(\Hat B_j\setminus\{h\})+\Normal v_i(h)\\
        &<1+1/3+2x=4/3-2x.
    \end{align*}
\end{proof}

\begin{lemma}\label{4/3ordered:lem:badbags}
    For all $B_j\in A^+\cup A^1$, we have $\Normal v_i(\Hat B_j)\leq 4/3+x$.
\end{lemma}
\begin{proof}
    Suppose $B_j\in A^+$. Then $j\leq \ell^*$, so
    \begin{align*}
        \Normal v_i(\Hat B_j)&=\Normal v_i(B_j)\\
        &=\Normal v_i(j)+\Normal v_i(2n-j+1)\\
        &\leq1+\Normal v_i(2n-\ell^*)&(2n-j+1\geq 2n-\ell^*)\\
        &=4/3+x.
    \end{align*}
    Now, if $B_j\in A^1$, let $g$ be the good added to $B_j$ in the bag filling phase. Then
    \begin{align*}
        \Normal v_i(\Hat{B_j})&=\Normal v_i(B_j)+\Normal v_i(g)\\
        &< 1+\Normal v_i(2n-\ell^*)&(\Normal v_i(B_j)<1 \text{ and } g\geq 2n-\ell^*)\\
        &=4/3+x.
    \end{align*}
\end{proof}

Now let $|A^1|=2n/3+\ell$, hence $|A^2|=n/3-(\ell+\ell^*)$. Let $P$ be a $\MMS_i^d$ partition (with respect to $\Normal v_i$), and let $R$ be the limitation of $P$ to $\{1, \dots,\allowbreak 8n/3 + \ell\}$. Without loss of generality, $|R_{1}|\geq\dots\geq|R_{4n/3}|$, and let $t$ be the number of bags of size 1 in $R$ that contain exactly one item from $\{1,\dots,\ell^*\}$ and nothing from $\{\ell^*+1, \dots,\allowbreak 8n/3 + \ell\}$.

\begin{lemma}\label{4/3ordered:lem:bagsizes}
    \[\sum_{j=1}^{t+\ell}|R_j|\geq 3(t+\ell)\]
\end{lemma}
\begin{proof}
    If $|R_{t+\ell}|\geq3$, the claim holds, since the bags are ordered by size. If not, we have
    \begin{align*}
        8n/3+\ell&=\sum_{j=1}^{4n/3}|R_j|\\
        &=\sum_{j=1}^{t+\ell}|R_j|+\sum_{j=t+\ell+1}^{4n/3-t}|R_j|+\sum_{j=4n/3-t+1}^{4n/3}|R_j|\\
        &\leq \sum_{j=1}^{t+\ell}|R_j| + |R_{t+\ell}|(4n/3-2t-\ell)+t\\
        &\leq \sum_{j=1}^{t+\ell}|R_j| +8n/3-3t-2\ell,
    \end{align*}
    hence $\sum_{j=1}^{t+\ell}|R_j|\geq 3(t+\ell)$
\end{proof}

\begin{lemma}\label{4/3ordered:lem:lem8}
    $\ell+\ell^*+t\leq 4n/3$.
\end{lemma}
\begin{proof}
    We have $\ell+2n/3+\ell^*=|A^+|+|A^1|\leq n$, and $t\leq\ell^*\leq n$, so $\ell+\ell^*+t\leq 4n/3$.
\end{proof}

\begin{lemma}\label{4/3ordered:lem:lem9}
    \[\sum_{j=2n-\ell^*+1}^{2n-t}\Normal v_i(j)+\sum_{j=8n/3-2\ell-t-2\ell^*+1}^{8n/3-2\ell-2t-\ell^*}\Normal v_i(j)<\ell^*-t.\]
\end{lemma}
\begin{proof}
    We are summing over a total of $2(\ell^*-t)$ items, each of which are less valuable than $1/2$ by Lemma \ref{4/3ordered:lem:bagsizes}. Hence the claim holds.
\end{proof}

\begin{lemma}\label{4/3ordered:lem:lem10}
    \[\sum_{j=t+1}^{\ell^*}\Normal v_i(j)+\sum_{j=8n/3-2\ell-2t-\ell^*+1}^{8n/3+\ell}\Normal v_i(j)\leq \ell+\ell^*\]
\end{lemma}
\begin{proof}
    By construction, there are $t$ bags in $R$ that contain exactly one item, and that item is from $\{1,\dots,\ell^*\}$. For each of these bags containing $j$ for $j\in\{t+1,\dots,\ell^*\}$, there exists a bag $B$ containing at least two items, one of them from $\{1,\dots,t\}$. By swapping $j$ with this item, the value of $B$ only decreases, so we can repeat this procedure until:
    \begin{enumerate}
        \item $\Normal v_i(B)\leq 1$ for all $B\in R$,
        \item All items from $\{1,\dots,t\}$ are alone in their bag, and
        \item All items from $\{t+1,\dots,\ell^*\}$ are in bags of size at least 2.
    \end{enumerate}
    Let $T_j\in R$ be the bag containing item $j$. Now consider the set $S=\{R_1,\dots,R_{t+\ell}\}\cup\{T_{t+1},\dots,T_{\ell^*}\}$. In case the two sets comprising $S$ intersect, add the lowest index $R_j$ to $S$ until $|S|=\ell+\ell^*$. 

    We now show that $S$ contains at least $3\ell+2\ell^*+t$ many items: By Lemma \ref{4/3ordered:lem:bagsizes}, $\sum_{j=1}^{t+\ell}|R_j|\geq 3(t+\ell)$. If the remaining $\ell^*-t$ many bags in $S\setminus\{R_1,\dots,R_{\ell+t}\}$ all contain at least 2 items, we are done. Otherwise all bags not in $S$ contain only one good, since all the $T_{t+1},\dots,T_{\ell^*}$ contain at least two items, and we added the remaining bags in decreasing order of size. Hence:
    \begin{align*}
        \sum_{B\in S}|B|&= 8n/3+\ell-\sum_{B\notin S}|B|\\
        &\geq 8n/3+\ell-(4n/3-\ell-\ell^*)\\
        &=4n/3+2\ell+\ell^*\\
        &\geq 3\ell+2\ell^*+t &(\text{Lemma \ref{4/3ordered:lem:lem8}}).
    \end{align*}
    Now recall that all bags in $S$ have value at most 1. Also, $S$ contains the goods $\{t+1,\dots,\ell^*\}$, and at least $3\ell+2t+\ell^*$ other goods. The $3\ell+2t+\ell^*$ lowest value goods that can appear in $R$ are $\{8n/3-2\ell-2t-\ell^*+1,\dots,8n/3+\ell\}$, so
    \begin{align*}
        \ell^*+\ell&\geq \sum_{B\in S}\Normal v_i(B)\\
        &\geq \Normal v_i(\{t+1,\dots,\ell^*\})\\&+\Normal v_i(\{8n/3-2\ell-2t-\ell^*+1,\dots,8n/3+\ell\}).
    \end{align*}
\end{proof}

\begin{lemma}\label{4/3ordered:lem:lem6}
    \begin{align*}
        \Normal v_i(&\{8n/3-2\ell-t-2\ell^*+1,\dots,8n/3+\ell\}\\&\cup\{t+1,\dots,\ell^*\}\\&\cup\{2n-\ell^*+1,\dots,2n-t\})\leq 2\ell^*+\ell-t,&
    \end{align*}
    and all unions above are disjoint unions.
\end{lemma}
\begin{proof}
    We notice that Lemmas \ref{4/3ordered:lem:lem9} and \ref{4/3ordered:lem:lem10} together account for each item considered here exactly once. Hence
    \begin{align*}
        \Normal v_i(&\{8n/3-2\ell-t-2\ell^*+1,\dots,8n/3+\ell\}\\&\cup\{t+1,\dots,\ell^*\}\\&\cup\{2n-\ell^*+1,\dots,2n-t\})\\
        &<(\ell^*-t)+(\ell^*+\ell)\\
        &=2\ell^*+\ell-t.
    \end{align*}
    Furthermore, since $2n/3+\ell+\ell^*=|A^1|+|A^+|\leq n$, we have
    \begin{align*}
        8n/3-2\ell^*-t-2\ell+1 &\geq 8n/3-2n/3-t+1\\&>2n-t,&
    \end{align*}
    so the first and second term are disjoint. All other terms are clearly pairwise disjoint.
\end{proof}

\PartOne*
\begin{proof}
    Assume towards a contradiction that agent $i$ does not get assigned to a bag, and that $\Normal v_i(2n-\ell^*)\geq 1/3$, as we have done in this section. Then by Lemma \ref{4/3ordered:lem:goodbags}, the $n/3-\ell-\ell^*$ many bags in $A^2$ have value less than $4/3-2x$. By Lemma \ref{4/3ordered:lem:badbags}, all other bags have value at most $4/3+x$. We now show that we can improve this bound on the bags \[S=\{\Hat B_1,\dots, B_{2\ell+t+2\ell^*-2n/3}\}\cup\{\Hat B_{t+1},\dots,\Hat B_{\ell^*}\}\cup\Hat A^1,\]
    where $\Hat A^1$ is the set of bags from $A^1$ after the algorithm terminates. Notice that, again since $\ell+\ell^*\leq n/3$, we have $2\ell+t+2\ell^*-2n/3< t+1\leq \ell^*$, and so al three terms in $S$ are disjoint and contained in $A^+\cup A^1$. We now differentiate two cases, which work broadly similar, but have some small differences:\\

    \textbf{Case 1:} $2\ell+t+2\ell^*\geq2n/3$. In this case, all three components of $S$ are nonempty, and $|S|=3(\ell+\ell^*)$. We will show that $\Normal v_i(\bigcup_{B\in S}B)\leq 4(\ell+\ell^*)$: 
    \begin{align*}
        \Normal v_i(\bigcup_{B\in S}B)&\leq \Normal v_i(\bigcup_{B\in A^1}B)\\
        &\quad+\Normal v_i(\{2n+1,\dots,8n/3+\ell\}\\
        &\quad+\Normal v_i(\{1,\dots,2\ell+t+2\ell^*-2n/3\})\\
        &\quad+\Normal v_i(\{8n/3-2\ell-t-2\ell^*+1,\dots,2n\})\\
        &\quad+\Normal v_i(\{t+1,\dots,\ell^*\})\\
        &\quad+\Normal v_i(\{2n-\ell^*+1,\dots,2n-t\}).&
    \end{align*}
    Now for all $B\in A^1$ we have $\Normal v_i(B)<1$, since otherwise there would be no item added to $B$ in the bag-filling phase, and $B$ would belong to $A^2$. Hence
    \[\Normal v_i(\bigcup_{B\in A^1}B)<2n/3+\ell.\]
    Also, since all goods have value at most 1 by normalization, 
    \[\Normal v_i(\{1,\dots,2\ell+t+2\ell^*-2n/3\})\leq 2\ell+t+2\ell^*-2n/3.\]
    The rest of the terms in the above inequality are covered exactly by Lemma \ref{4/3ordered:lem:lem6}. Adding these inequalities gives
    \[\Normal v_i(\bigcup_{B\in S}B)\leq 4(\ell+\ell^*).\]
    We note that we now only need to bound $2n/3-2\ell-2\ell^*$ many bags using Lemma \ref{4/3ordered:lem:lem6}. Taking everything together, we get
    \begin{align*}
        \Normal v_i(M)&<(n/3-\ell-\ell^*)(4/3-2x)\\&\quad+(2n/3-2\ell-2\ell^*)(4/3+x)\\&\quad+4(\ell+\ell^*)\\&=4n/3,
    \end{align*}
    which is a contradiction to $\Normal v_i(M)=4n/3$.\\

    \textbf{Case 2:} $2\ell+t+2\ell^*=2n/3-\beta$ for $\beta>0$. We argue similarly to the first case, but pay special attention to the fact that the first component of $S$ is empty. In this case, since $2\ell+t+2\ell^*-2n/3+\beta=0$, we have
    \begin{align*}
        |S|&=\ell^*-t+2n/3+\ell\\&=\ell^*-t+2n/3+\ell-(2\ell+t+2\ell^*-2n/3+\beta)\\&=3(\ell+\ell^*)+\beta.
    \end{align*}
    The total value of these goods is bounded by
    \begin{align*}
        \Normal v_i(\bigcup_{B\in S}B)&\leq \Normal v_i(\bigcup_{B\in A^1}B)\\
        &\quad+\Normal v_i(\{2n+1,\dots,8n/3+\ell\}\\
        &\quad+\Normal v_i(\{t+1,\dots,\ell^*\})\\
        &\quad+\Normal v_i(\{2n-\ell^*+1,\dots,2n-t\}).&
    \end{align*}
    The first term again follows
    \[\Normal v_i(\bigcup_{B\in A^1}B)<2n/3+\ell,\]
    and Lemma \ref{4/3ordered:lem:lem6} covers everything except $\{2n+1,\dots,2n+\beta\}$, since $8n/3-2\ell-t-2\ell^*=2n+\beta$. The uncovered goods $g$ are each of value $\Normal v_i(g)\leq \Normal v_i(2n-\ell^*)=1/3+x$, so their total value is no more than $\beta(1/3+x)$. Adding these inequalities gives
    \begin{align*}
        \Normal v_i(\bigcup_{B\in S}B)&\leq 2n/3+\ell+2\ell^*+\ell-t+\beta(1/3+x)\\
        &=4(\ell+\ell^*)+\beta(4/3+x),
    \end{align*}
    again since $2\ell+t+2\ell^*-2n/3+\beta=0$. In this case, we only need to bound $2n/3-2\ell-2\ell^*-\beta$ many bags with Lemma \ref{4/3ordered:lem:badbags}, as $|S|=3(\ell+\ell^*)+\beta$, so we get
    \begin{align*}
        \Normal v_i(M)&<(n/3-\ell-\ell^*)(4/3-2x)\\&\quad+(2n/3-2\ell-2\ell^*-\beta)(4/3+x)\\
        &\quad+4(\ell+\ell^*)+\beta(4/3+x)\\&=4n/3,&
    \end{align*}
    which again is in contradiction to $\Normal v_i(M)=4n/3$
\end{proof}

\theoremThree*

\begin{proof}
    Theorems \ref{4/3ordered:thm:part1} and \ref{4/3ordered:thm:part2} together with Lemma \ref{4/3ordered:lem:EF} show that the partial allocation resulting from Algorithm \ref{4/3ordered:alg:main} is 1-out-of-$4\ceil{n/3}$ MMS and EF1. In order to complete this allocation, we turn back to Lemma \ref{3/2ordered:lem:envy_cycles}. Since we do not only need to allocate the goods that were possibly leftover after Algorithm \ref{4/3ordered:alg:main} concludes, but also the goods that were allocated to dummy agents, we cannot use the stronger version that promises EFX. Instead, we can only maintain the EF1 guarantee by disregarding the ordering.
    
    In conclusion, inputting the result from Algorithm \ref{4/3ordered:alg:main} into the envy-cycle elimination procedure from Lemma \ref{3/2ordered:lem:envy_cycles} yields a complete allocation that is 1-out-of-$4\ceil{n/3}$ MMS and EF1 at the same time.
\end{proof}

\end{document}